\documentclass[conference]{IEEEtran}
\usepackage{amssymb}
 \usepackage{graphicx}
 \usepackage{amsmath}
 \usepackage{algorithm}
\usepackage{algpseudocode}
\usepackage{amsthm}

\newtheorem*{problem}{Problem}
\newtheorem{theorem}{Theorem}

\ifCLASSINFOpdf
\else
\fi

\begin{document}
%\title{Age of Entanglement: Continuous Distribution in Satellite Repeater Chains with Intermittent Availability}

\title{Age of Entanglement in Satellite Repeater Chains with Intermittent Availability}

%\author{\IEEEauthorblockN{Michael Shell}
%\IEEEauthorblockA{School of Electrical and\\Computer Engineering\\
%Georgia Institute of Technology\\
%Atlanta, Georgia 30332--0250\\
%Email: http://www.michaelshell.org/contact.html}
%\and
%\IEEEauthorblockN{Homer Simpson}
%\IEEEauthorblockA{Twentieth Century Fox\\
%\and
%\IEEEauthorblockN{James Kirk\\ and Montgomery Scott}
%\IEEEauthorblockA{Starfleet Academy\\
%San Francisco, California 96678--2391\\


\author{\IEEEauthorblockN{Elif Tugce Ceran
\IEEEauthorblockA{Department of Electrical and Electronics Engineering\\
Faculty of Engineering, Middle East Technical University\\
Ankara, Turkey\\
elifce@metu.edu.tr}
%Telephone: (800) 555--1212, Fax: (888) 555--1212}
%\IEEEauthorblockA{\IEEEauthorrefmark{4}Tyrell Inc., 123 Replicant Street, Los Angeles, California 90210--4321}
}}


\maketitle

\begin{abstract}

Timely availability of high-fidelity entanglement
is essential for emerging quantum networks. This paper
introduces the Age of Entanglement (AoE) as a  novel performance
metric that captures the freshness of bipartite entanglement
under continuous distribution in quantum repeater chains. AoE
extends classical Age of Information (AoI) based metrics to
quantum networking by capturing storage, decoherence,
probabilistic entanglement generation and swapping.

We study a satellite-assisted quantum repeater network in which entangled pairs are generated probabilistically, stored in quantum memories suffer from decoherence, and combined to form end-to-end entangled links. Satellite–ground connectivity is intermittent and modeled as a two-state Markov chain. The resulting AoE minimization problem is formulated as an infinite-horizon Markov decision process (MDP), where control actions determine when to generate, store, or swap entangled pairs under stochastic link availability and memory degradation.

Using relative value iteration, we characterize AoE-optimal policies and evaluate their performance numerically. Our results highlight the impact of decoherence, imperfect operations, and visibility dynamics, and show that the proposed dynamic policies significantly outperform swap-as-soon-as-possible and greedy entanglement generation strategies. Our results provide practical design and control guidelines for satellite-enabled quantum repeater chains supporting continuous entanglement distribution.

\end{abstract}


% \fi
\IEEEpeerreviewmaketitle

\section{Introduction}

The development of quantum communication networks is reshaping communication and computation by enabling secure quantum-enabled applications \cite{Quantum_general}. These networks support tasks such as quantum teleportation and quantum key distribution (QKD) \cite{Quantum_general,Micius,advances}. Bipartite entanglement constitutes the essential resource for these tasks, providing a quantum link between spatially separated nodes. However, photon loss in optical fibers limits the direct distribution of entanglement over long distances. To address this challenge, quantum repeaters and satellite-based nodes are employed to extend entanglement distribution to global scales.

%In contrast to classical communication, which transmits information using discrete bits over reliable channels, quantum communication relies on entanglement as its primary enabling resource. Entanglement distribution begins with probabilistic heralded entanglement generation (HEG), in which neighboring nodes attempt to establish a shared quantum state and success is confirmed via classical feedback. Long-distance entanglement is achieved through entanglement swapping, whereby intermediate nodes, such as quantum repeaters or satellites, perform Bell state measurements (BSMs) to combine shorter entangled links into an end-to-end connection without requiring qubits to traverse the full distance. However, quantum states stored in noisy memories are highly time-sensitive and experience decoherence, leading to a progressive degradation in fidelity. Since applications such as QKD impose strict minimum fidelity requirements, entangled links must be discarded once they become unusable.

Classical communication achieves reliability through coding and retransmission of bits over noisy channels, whereas quantum communication is fundamentally constrained by principles such as the no-cloning theorem. As a result, quantum communication often relies on shared entanglement between distant nodes. In repeater-based networks, entanglement is typically generated probabilistically between neighboring nodes and confirmed via classical signaling, then extended over longer distances through entanglement swapping at intermediate nodes. Since stored quantum states are highly sensitive to noise and decoherate over time, their quality degrades as they age.

Satellite-based quantum networks offer a distinct advantage for long-distance entanglement distribution by exploiting low-loss free-space optical channels \cite{Micius,advances,yang2026optimizing}. However, they are constrained by intermittent visibility and limited contact durations due to orbital dynamics and onboard resource constraints. Despite these advantages, entanglement distribution remains challenged by the time sensitivity of quantum states and  the inherently probabilistic nature of generating entanglement between neighboring nodes and combining shorter entangled links to form end-to-end links with swapping.

%the inherently probabilistic nature of operations such as such as heralded entanglement generation (HEG) and Bell state measurements (BSMs).

%To address this challenge, we adopt the Age of Entanglement (AoE) as a performance metric that captures the cumulative decoherence experienced by quantum links during generation and swapping. AoE provides a natural foundation for designing control policies that jointly account for probabilistic operations, memory aging, and intermittent connectivity in satellite-based quantum networks.

To quantify the freshness of quantum resources, we introduce the Age of Entanglement (AoE) as a performance metric, inspired by the classical Age of Information (AoI) \cite{Kaul2012,Yates2021}. In classical networks, AoI resets to the age of the most recently delivered update. Similarly, AoE resets upon successful generation of an end-to-end entangled link; however, it must additionally capture quantum-specific effects such as decoherence and imperfect entanglement generation and swapping, which accumulate while qubits are stored and processed. AoE  provides a natural framework for designing control policies that jointly account for probabilistic operations, memory aging, and intermittent connectivity in satellite-based quantum networks.

\subsection{Related Work}

%Entanglement distribution is typically categorized into on-demand and continuous distribution (CD) protocols \cite{kolar2022adaptive,inesta2023performance}. On-demand protocols establish links only after a specific request is triggered, making them more resource-efficient for networks with a limited number of qubits per node. In contrast, CD protocols continuously supply entangled states to the nodes, allowing background applications like QKD to consume resources without explicit scheduling. While on-demand distribution requires complex scheduling policies to prevent traffic congestion, CD protocols are expected to allocate resources faster in large-scale heterogeneous networks. Recent advancements have introduced adaptive continuous schemes that leverage information from previous requests to guide the random generation of quantum links, significantly reducing service latencies \cite{kolar2022adaptive,inesta2023performance}.

Entanglement distribution protocols are typically classified as on-demand or continuous distribution (CD) schemes \cite{kolar2022adaptive,inesta2023performance}. On-demand approaches \cite{haldar2024fast,yau2025reinforcement,Chehimi2025_delay} establish links only upon request and can incur scheduling overhead as networks scale. In contrast, CD protocols proactively supply entanglement, enabling background applications such as QKD and faster resource allocation in heterogeneous networks. Recent adaptive CD schemes further reduce service latency by leveraging information from prior requests \cite{kolar2022adaptive,inesta2023performance}.

%Satellite networks offer a distinct advantage for long-distance distribution by utilizing free-space propagation to reduce channel loss and decoherence. However, these systems are constrained by intermittent visibility and limited line-of-sight range, which can be modeled using 2-state Markov chains. 

Satellite-based quantum networks, such as entanglement distribution using hybrid links (ED-HL), combine satellite downlinks, inter-satellite links, and onboard quantum memories to enable global-scale entanglement distribution, with resource allocation typically formulated as an integer programming problem under link availability and resource constraints \cite{yang2026optimizing}.

Quantum network dynamics are well-suited for modeling as Markov decision processes (MDPs) \cite{puterman} and  previous work leveraged MDPs \cite{haldar2024fast,yau2025reinforcement,haldar2025reducing} and dynamic programming to optimize entanglement distribution policies. Recent research has also employed reinforcement learning (RL), such as Q-learning, to discover hardware-aware policies that outperform the standard swap-as-soon-as-possible (SWAP-ASAP) approach \cite{haldar2024fast} and address non-linear application-driven objectives \cite{yau2025reinforcement}.

On the other hand, Age of Information (AoI) based policies have demonstrated significant gains in semantic and goal-oriented classical communication, as AoI is a well-established metric for capturing the freshness of data updates while jointly accounting for delay and throughput \cite{Kaul2012,Yates2021}. Minimizing AoI has been widely modeled as a MDP, and state-dependent dynamic policies have been shown to be highly effective in prior work \cite{ceran2019average,Sun2020AoIBook}.

%Satellite-based quantum networks enable long-distance entanglement distribution by leveraging low-loss free-space optical links, but are fundamentally constrained by intermittent visibility, probabilistic hardware operations, and decoherence in quantum memories. As a result, entangled links are both unreliable to establish and highly time-sensitive once generated. While existing entanglement distribution protocols primarily focus on maximizing success probability or throughput, they do not explicitly control the freshness of stored entanglement under dynamic network conditions. This omission is critical, as stale entanglement may violate fidelity requirements for applications such as quantum key distribution, leading to wasted resources and failed sessions. 

To the best of our knowledge, while link ages have been employed as state variables and fidelity indicators in prior works, the Age of Entanglement (AoE) of the end-to-end link has not been explicitly formulated. Moreover, the relationship and fundamental differences between the classical AoI for data packets and the AoE for entangled quantum links have not been systematically characterized.

The main contributions of this paper are as follows:
\begin{itemize}
    \item We introduce a novel modeling framework that, for the first time, captures intermittent availability of repeater nodes in quantum entanglement distribution.
    \item We propose a state-dependent dynamic control policy for entanglement generation and swapping that significantly outperforms stationary heuristics such as swap-as-soon-as-possible and greedy regeneration.
    \item We formalize the Age of Entanglement (AoE) metric, inspired by the AoI, to jointly capture entanglement generation delay and quality degradation due to decoherence.
\end{itemize}

\section{System Model}

We consider a  discrete-time and slotted linear quantum repeater network for distributing bipartite entanglement between two ground stations, \emph{Alice} and \emph{Bob}, via satellite-based repeater nodes. A shared bipartite entangled state between two nodes is referred to as an entangled link. Satellites act as intermediate repeaters that relay entanglement between neighboring nodes to enable long-distance quantum communication. In this work, we focus on the single-repeater case, where one satellite (node 2) serves as the intermediate node between \emph{Alice} (node 1) and \emph{Bob} (node 3).

%\subsection{Network Topology and Connectivity}

%We consider a linear quantum repeater network consisting of $n$ nodes arranged in a chain. The network connects two ground stations, referred to as \emph{Alice} and \emph{Bob}, through one or more satellite-based repeater nodes. These satellites enable long-distance quantum communication by relaying entanglement between the ground terminals.

%In this work, we focus on two representative network configurations:

%\begin{itemize}
%    \item \textbf{Three-node chain:} A single satellite serves as an intermediate repeater between Alice and Bob. This configuration is applicable when both ground stations are simultaneously within the satellite's coverage footprint.
    
 %   \item \textbf{Four-node chain:} Two satellites act as intermediate repeater nodes and are connected via an inter-satellite link (ISL). This architecture enables end-to-end quantum communication between ground stations that do not share visibility with the same satellite.
%\end{itemize}

%\subsection{Intermittent Satellite Visibility}

%\subsection{Network Topology and Connectivity}

Satellite-to-ground connectivity is inherently intermittent due to orbital motion and elevation-angle constraints. A communication link between a ground station and a satellite is available only when the satellite exceeds a minimum elevation angle $\theta_e$, ensuring sufficient link quality. An elementary link is a heralded bipartite entangled state established via direct physical downlinks between a satellite source and each of its visible ground stations. A virtual end-to-end link is a logical connection generated between physically non-neighboring ground nodes through entanglement swapping operations performed by intermediate satellite repeater.

%We model the visibility of each ground--satellite link as a two-state Markov process with states $V \in \{0,1\}$, corresponding to \emph{invisible} and \emph{visible}, respectively. A link between node $i$ and satellite $j$ can be established at time $t$ only when the visibility indicator satisfies 
%\begin{equation}
%    v_{ij}(t) = 1 .
%\end{equation}
% $v^{12}_t \in \{0,1\}$ and $v^{23}_t \in \{0,1\}$ denote the visibility states associated with links $(1,2)$ and $(2,3)$, respectively. A value of $v^{ij}_t = 1$ indicates that the corresponding elementary link can be requested at time $t$, while $v^{ij}_t = 0$ indicates that the link is not visible.

 We model the visibility of each ground–satellite link as a two-state Markov process with state space \(V \in \{0,1\}\), where \(0\) and \(1\) denote the invisible and visible states, respectively. An elementary entangled link between node \(i\) and satellite \(j\) can be established at time slot \(t\) only if $v_t^{ij} = 1$. Let \(v_t^{12} \in \{0,1\}\) and \(v_t^{23} \in \{0,1\}\) denote the visibility states of links \((1,2)\) and \((2,3)\), respectively. When \(v_t^{ij} = 1\), the corresponding elementary link is available and nodes can request elementary link generation at time \(t\); otherwise, it is unavailable.

%In the four-node configuration, we adopt a hybrid entanglement distribution strategy (ED-HL). Successful end-to-end operation in this case requires simultaneous visibility of both ground--satellite links, as well as an unobstructed line-of-sight inter-satellite link between the two repeater satellites.

%We consider a single repeater chain of N equidistant and identical nodes, which could be a part of a larger quantum network. 

To provide high fidelity and timely end-to-end entanglement link between the two end nodes, we assume the nodes can perform the following operations at each time slot: i) heralded generation of entanglement between each ground node and intermittently available satellite, which succeeds with probability $p_L$, and discarding any link that existed before attempt, ii) entanglement swap attempts, which consumes two adjacent entangled links to generate a longer distance link between ground nodes with probability $p_s$ and  iii) waiting for a better opportunity to generate or swap in the future and letting the links decohere. We also adopt a common assumption and assume discarding any link that existed for some time denoted by $m^*$ and emptying qubit resources to prevent generation of low-quality end-to-end entanglement due to decoherence.

In real quantum memories, the stored entangled state degrades over time due to decoherence, and its fidelity $F(m)$, defined as the closeness between the noisy state and the ideal state, typically decays with storage time $m$; for example, in simple phase-damping models the fidelity of a stored entangled pair decreases approximately exponentially as
\begin{equation}
F(m) \approx e^{-m/T_2},
\end{equation}
where $T_2$ is the memory coherence time that characterizes how rapidly phase information is lost. Consequently, larger link ages correspond to less fidelity, which motivates the use of age-based policies to avoid poor-quality end-to-end entanglement as commonly adopted in previous work~\cite{haldar2024fast}.

Let $m^{12}_t$ and $m^{23}_t$ denote the ages of the elementary links $(1,2)$ and $(2,3)$ at time $t$, respectively. Each elementary-link age takes values in the set $\{-1,1,2,\ldots,m^\star\}$, where $-1$ represents an inactive (non-existent) link and positive integers represent the number of time slots elapsed since the entanglement was generated. Newly generated elementary links are initialized with age equal to~1; consequently, age~0 never occurs. If an elementary link age exceeds the cutoff $m^\star$, the corresponding entanglement is discarded and the age is reset to $-1$.

%We employ relative value iteration (RVI) to find optimal policies and demonstrate that our propose policies outperforms swap-asap-policy and optimal policies with deterministic cut-off times. 

\subsection{Age of Entanglement (AoE)}

We introduce the \emph{Age of Entanglement (AoE)} as a performance metric that captures the freshness of the most recently generated end-to-end entangled pair between nodes~1 and~3. 

Note that, unlike the Age of Information (AoI) used for elementary links, the AoE is not associated with the presence or absence of a physical quantum state, but instead records the time elapsed since the last successful creation of an end-to-end entangled link. The AoE captures both the elapsed time since the most recent successful end-to-end entanglement generation and the degradation of the resulting entangled state due to aging and decoherence, whereas the AoI is a classical metric used to quantify the freshness of information packets.

Let ${\Delta}^e_t$ denote the AoE at the beginning of the time slot~$t$. The AoE evolves stochastically with time depending on the control actions and successful swap probabilities and is updated as follows:
\begin{itemize}
\item If an entanglement swapping operation is successfully performed at time~$t$ with probability $p_s$, producing an end-to-end entangled pair using elementary links with ages $m^{12}_t$ and $m^{23}_t$, then the AoE is \emph{reset} according to
\begin{equation}
{\Delta}^e_{t+1} = m^{12}_t + m^{23}_t.
\end{equation}
This update reflects the standard entanglement swapping formalism in quantum repeaters \cite{haldar2024fast}, where the swapped state inherits the accumulated storage induced decoherence of the  elementary links.
%provided that $m^{12}_t + m^{23}_t \le 2m^\star$.
\item If no successful end-to-end entanglement is generated at time~$t$, then the AoE increases by one unit:
\begin{equation}
{\Delta}^e_{t+1} = {\Delta}^e_t + 1.
\end{equation}
\end{itemize}

The AoE thus represents the age of the most recent end-to-end entanglement, regardless of whether an entangled pair is currently available. Between successive successful swapping events, the AoE grows linearly with time, and it is reset to a value that reflects the freshness of the elementary links used to generate the end-to-end entanglement.

This formulation differs fundamentally from terminal or episodic models, as the system continues to operate after each successful end-to-end entanglement generation. Consequently, AoE provides a natural objective for continuing-control formulations, allowing policies to trade off the frequency of successful entanglement generation (or end-to-end entanglement generation delay) against the freshness of the resulting end-to-end links.

\begin{problem}
    The objective is to determine a stationary policy $\pi$ that minimizes the infinite horizon expected average AoE
\begin{equation}
\eta^\pi \triangleq
\liminf_{T \to \infty}
\frac{1}{T}
\mathbb{E}_{\pi}
\left[
\sum_{t=0}^{T-1}
{\Delta}^e_{t+1}
\right].
\end{equation}
\end{problem}

\section{MDP Formulation and RVI}
\subsection{Markov Decision Process Formulation}

We model the operation of a three-node linear quantum repeater chain as a discrete-time average reward Markov decision process (MDP) with tuple $<\mathcal{S},\mathcal{A},P, r>$ \cite{puterman}. The network consists of nodes $1$--$2$--$3$, where entanglement can be generated on the elementary links $(1,2)$ and $(2,3)$, and an end-to-end virtual link $(1,3)$ can be established via entanglement swapping at the intermediate node~2.

\subsubsection{State Space}
At time slot $t$, the system state is
\begin{align}
s_t \triangleq \big( v^{12}_t,\; v^{23}_t,\; m^{12}_t,\; m^{23}_t,\; {\Delta}^e_t \big) \in \mathcal{S},
\end{align}
where $v^{12}_t \in \{0,1\}$ and $v^{23}_t \in \{0,1\}$ denote the visibility states associated with links $(1,2)$ and $(2,3)$, respectively. A value of $v^{ij}_t = 1$ indicates that the corresponding elementary link can be requested at time $t$, while $v^{ij}_t = 0$ indicates that the link is not visible.

The variable ${\Delta}^e_t$ denotes the AoE of the end-to-end link $(1,3)$, representing the age of the most recently generated end-to-end entangled pair. The AoE takes values in $\{1,2,\ldots,\Delta_{\max}\}$, and ${\Delta}^e_t$ increases by one every time slot until the next successful end-to-end entanglement generation. To ensure a finite state space, the AoE is capped at a large constant $\Delta_{\max}$, i.e., once ${\Delta}^e_t$ reaches $\Delta_{\max}$ it remains at $\Delta_{\max}$ until the next successful reset. 

%To ensure a finite state space, the AoE is capped at a large constant $\Delta_{\max}$, such that
%\begin{equation}
%{\Delta}^e_{t+1} = \min\bigl({\Delta}^e_{t+1},\; \Delta_{\max}\bigr).
%\end{equation}

In contrast to episodic or on-demand entanglement generation policies, the MDP is continuing, that is, there are no terminal or absorbing states, and the process evolves indefinitely.

\subsubsection{Action Space}

At each time slot $t$, the controller selects an action
\begin{equation}
a_t \triangleq \big( a^{12}_t,\; a^{23}_t,\; a^{\mathrm{sw}}_t \big) \in \mathcal{A},
\end{equation}
where $a^{12}_t \in \{0,1\}$ and $a^{23}_t \in \{0,1\}$ indicate whether a request is made to generate entanglement on links $(1,2)$ and $(2,3)$, respectively, and $a^{\mathrm{sw}}_t \in \{0,1\}$ indicates whether an entanglement swapping operation is attempted at node~2.

Not all actions are admissible in every state, i.e., $a_t \in \mathcal{A}(s)$, where  $\mathcal{A}(s)$ represents the admissible policies for each state $s_t=s$, Specifically, entanglement swapping cannot be performed simultaneously with link generation, that is, if $a^{\mathrm{sw}}_t = 1$, then $a^{12}_t = a^{23}_t = 0$.

Moreover, link generation requests are gated by visibility, such that $a^{12}_t = 1$ requires $v^{12}_t = 1$, and $a^{23}_t = 1$ requires $v^{23}_t = 1$. Finally, a swap action is feasible only if both elementary links exist, i.e., $a^{\mathrm{sw}}_t = 1$ requires $m^{12}_t \ge 1$ and $m^{23}_t \ge 1$. 

%The set of admissible actions in state $s$ is denoted by $\mathcal{A}(s)$.

\subsubsection{Visibility Dynamics}

The satellite visibility processes evolve exogenously and independently of the control actions. Specifically, $v^{12}_t$ and $v^{23}_t$ each follow a two-state Markov chain with transition probabilities
\begin{align}
\Pr\!\left( v^{12}_{t+1} = j \mid v^{12}_t = i \right) &= [\mathbf{P}_{12}]_{i,j}, \nonumber \\
\Pr\!\left( v^{23}_{t+1} = j \mid v^{23}_t = i \right) &= [\mathbf{P}_{23}]_{i,j}, \label{eq:visibility}
\end{align}
for $i,j \in \{0,1\}$, where $\mathbf{P}_{12}$ and $\mathbf{P}_{23}$ are given $2 \times 2$ stochastic matrices.

%\subsubsection{Elementary-Link Aging and Memory Cutoff}

%For $\ell \in \{12,23\}$, define the elementary-link aging operator
%\begin{equation}
%g(m) \triangleq
%\begin{cases}
%m+1, & \text{if } m \in \{1,\ldots,m^\star\}\ \text{and}\ m+1 \le m^\star,\\
%-1, & \text{if } m \in \{1,\ldots,m^\star\}\ \text{and}\ m+1 > m^\star,\\
%-1, & \text{if } m=-1.
%\end{cases}
%\end{equation}
%Thus, an active elementary link ages by one each time slot and is discarded if its age exceeds $m^\star$.

\subsubsection{State Transitions}

The one-step transition probability $P(s_{t+1} \mid s_t, a_t)$ factorizes into visibility transition probabilities as in \eqref{eq:visibility} and link-state transitions as follows.
%\begin{equation}
%\Pr(s_{t+1} \mid s_t, a_t)
%=
%\Pr(v^{12}_{t+1} \mid v^{12}_t)\,
%\Pr(v^{23}_{t+1} \mid v^{23}_t)\,
%\Pr(m^{12}_{t+1}, m^{23}_{t+1}, {\Delta}^e_{t+1} \mid s_t, a_t).
%\end{equation}

%\paragraph{(i) Swap action ($a^{\mathrm{sw}}_t=1$).}
When a swap is attempted, $a^{\mathrm{sw}}_t=1$, both parent links are consumed regardless of success:
\begin{equation}
m^{12}_{t+1} = m^{23}_{t+1} = -1.
\end{equation}
With probability $p_{\mathrm{sw}}$, the swap succeeds and the AoE is reset to the sum of link ages $m^{12}_t + m^{23}_t$; otherwise increases by one in every time slot, after applying the swap-reset rule, and is capped at $\Delta_{\max}$. 
\begin{equation}
{\Delta}^e_{t+1} =
\begin{cases}
m^{12}_t + m^{23}_t, & \text{w.p. } p_{\mathrm{sw}} \text{ if } a^{sw}_t=1,\\
\min\bigl( {\Delta}^e_t+1,\ \Delta_{\max} \bigr), & \text{otherwise}.
\end{cases}
\end{equation}
If no swap is attempted, $a^{\mathrm{sw}}_t=0$, elementary link states are determined by
\begin{equation}
m^{12}_{t+1} = g(m^{12}_t), \qquad  m^{23}_{t+1} = g(m^{23}_t),
\end{equation}
where $g(m)$ denotes the elementary-link aging operator for $\ell \in \{12,23\}$.
\begin{equation}
g(m) \triangleq
\begin{cases}
m+1, & \text{if } m \in \{1,\ldots,m^\star\}\ \text{and}\ m+1 \le m^\star,\\
-1, & \text{if } m \in \{1,\ldots,m^\star\}\ \text{and}\ m+1 > m^\star,\\
-1, & \text{if } m=-1.
\end{cases}
\end{equation}
Thus, an active elementary link ages by one each time slot and is discarded if its age exceeds $m^\star$.

Regeneration follows a drop-then-attempt rule. If $a^{12}_t = 1$, the $(1,2)$ link is dropped and replaced by
\begin{equation}
m^{12}_{t+1} =
\begin{cases}
1, & \text{w.p. } p_L,\\
-1, & \text{w.p. } 1-p_L,
\end{cases}
\end{equation}
independently of $m^{12}_t$. If $a^{12}_t=0$, then $m^{12}_{t+1} = g(m^{12}_t)$. Similarly for $m^{23}_{t+1}$ with $a^{23}_t$. When both links are requested simultaneously, the success events are conditionally independent.

%In this case (no swap), we set $R_t = {\Delta}^e_t$.

%\subsubsection{AoE Evolution}

%The AoE increases by one in every time slot, after applying the swap-reset rule, and is capped at $\Delta_{\max}$:
%\begin{equation}
%{\Delta}^e_{t+1} = \min\bigl( R_t + 1,\ \Delta_{\max} \bigr).
%\end{equation}

\subsubsection{Reward}

The immediate reward associated with the transition from $s_t$ to $s_{t+1}$ under action $a_t$ is defined as
\begin{equation}
r(s_t,a_t,s_{t+1}) = -{\Delta}^e_{t+1}.
\end{equation}
Thus, maximizing reward is equivalent to minimizing the long-term average AoE.  

%The objective is to determine a stationary policy $\pi : \mathcal{S} \to \mathcal{A}$ that maximizes the infinite horizon expected average reward, which equivalently minimizes the expected average AoE.

Then, the objective is to determine a stationary policy $\pi : \mathcal{S} \to \mathcal{A}$ that maximizes the infinite horizon expected average reward
\begin{equation}
\eta^\pi \triangleq
\liminf_{T \to \infty}
\frac{1}{T}
\mathbb{E}_{\pi}
\left[
\sum_{t=0}^{T-1}
r(s_t,a_t,s_{t+1})
\right], \label{average_reward}
\end{equation}
which equivalently minimizes the expected average AoE.

Next, Theorem 1 establishes that the optimal stationary policy $\pi^*$ solving the average-reward MDP in \eqref{average_reward} can be obtained via Relative Value Iteration (RVI) \cite{puterman}, as summarized in Algorithm~\ref{algorithm_RVI}.

\begin{theorem}[Optimal stationary policy and RVI optimality]
For the satellite repeater-chain average AoE minimization problem defined in \eqref{average_reward}, there exists an average-cost optimal stationary deterministic policy $\pi^*$. Moreover, the Relative Value Iteration (RVI) algorithm converges (up to an additive constant), and the resulting stationary policy is optimal.
\end{theorem}
\begin{proof}
The state of the system $s_t = (v^{12}_t, v^{23}_t, m^{12}_t, m^{23}_t, \Delta^e_t)$ takes values in finite sets; hence the induced average-cost MDP is finite-state and finite-action. Consider the  state 
$ (0,0,-1,-1,\Delta_{\max}).$ Because the link visibilities are ergodic two-state Markov chains, both links remain invisible for  $m^\star + (\Delta_{\max}-1)$ consecutive slots with positive probability. During this period, no new links can be generated, existing links are discarded within $m^\star$ slots, and $\Delta^e_t$ reaches $\Delta_{\max}$. Thus, $(0,0,-1,-1,\Delta_{\max})$ is reachable from every state with positive probability under any stationary policy, implying that the MDP is unichain. 

For finite-state unichain average-cost MDPs, standard results (\cite{puterman}, Ch.~8) guarantee the existence of an optimal stationary  policy and convergence of RVI, completing the proof.
\end{proof}

\begin{algorithm}[t]
\caption{Relative Value Iteration (RVI) for AoE Minimization in a Quantum Repeater}
\label{alg:rvi_aoe}
\begin{algorithmic}[1]
\Require
State space $\mathcal S$ with states 
$s=(v^{12},v^{23},m^{12},m^{23},\Delta^e)$;
admissible action sets $\mathcal A(s)$;
transition kernel $P(\cdot\mid s,a)$ induced by link aging, regeneration, swapping, and visibility dynamics;
AoE-based reward function $r(s,a,s')=-\Delta^e$;
reference state $s_{\mathrm{ref}}$;
tolerance $\varepsilon$.
\Ensure Optimal stationary policy $\pi^\star$

\State Initialize relative value function $h^{(0)}(s)\gets 0$ for all $s\in\mathcal S$
\State Initialize policy $\pi^{(0)}(s)$ arbitrarily
\State Set iteration index $k\gets 0$

\Repeat
    \State $k \gets k+1$
    \ForAll{$s \in \mathcal S$}
        \ForAll{$a \in \mathcal A(s)$}
            \State Compute
            \[
            Q^{(k)}(s,a)
            =
            \sum_{s' \in \mathcal S}
            P(s' \mid s,a)
            \bigl[
            -\Delta^{e\prime}(s') + h^{(k-1)}(s')
            \bigr],
            \]
            where $-\Delta^{e\prime}(s')$ denotes the AoE component of the next state $s'$.
        \EndFor
        \State Update relative value:
        \[
        h^{(k)}(s) \gets \max_{a \in \mathcal A(s)} Q^{(k)}(s,a).
        \]
        \State Update policy:
        \[
        \pi^{(k)}(s) \gets \arg\max_{a \in \mathcal A(s)} Q^{(k)}(s,a).
        \]
    \EndFor
    \State \textbf{Normalization:}
    \[
    h^{(k)}(s) \gets h^{(k)}(s) - h^{(k)}(s_{\mathrm{ref}}),
    \qquad \forall s\in\mathcal S.
    \]
\Until{$\max_{s \in \mathcal S} \lvert h^{(k)}(s) - h^{(k-1)}(s) \rvert < \varepsilon$}

\State \Return $\pi^\star \gets \pi^{(k)}$
\end{algorithmic}
\label{algorithm_RVI}
\end{algorithm}

\section{Simulation Results}

This section evaluates the performance of the proposed RVI policy and compares it against two baseline strategies, namely the \emph{Greedy Gen+Swap ASAP} policy and the \emph{Wait-Until-Ready (WUR)} policy, which are commonly used in the quantum repeater literature
due to their simplicity and intuitiveness. We focus on the average AoE as the primary performance metric and examine its behavior under varying elementary link generation probabilities $p_L$ and swap success probabilities $p_{sw}$ under different visibility scenarios.

\subsection{Benchmark Heuristic Policies}

%To evaluate the performance of the proposed RVI-based control policy, we consider the following benchmark heuristics, which are commonly used in the quantum repeater literature due to their simplicity and intuitive appeal.

%\subsubsection{Greedy Generation and Swap-ASAP Policy}

The \emph{Greedy Generation + Swap-ASAP} policy attempts to maximize immediate progress toward end-to-end entanglement. At each time slot, the controller behaves as follows:
\begin{itemize}
    \item If both elementary links $(1,2)$ and $(2,3)$ are present, an entanglement swap is attempted immediately, that is, $a^{sw}=1$.
    \item Otherwise, the controller requests entanglement generation on all elementary links whenever the satellite is visible, regardless of whether an existing link is already present or not.
\end{itemize}
%Formally, the greedy policy selects the action
%\[
%\pi_{\mathrm{greedy}}(s) =
%\begin{cases}
%(0,0,1), & \text{if } \Delta^{12} \ge 1 \text{ and } \Delta^{23} \ge 1, \\
%(\mathbbm{1}\{v^{12}=1\},\, \mathbbm{1}\{v^{23}=1\},\,0), & \text{otherwise}.
%\end{cases}
%\]

%\subsubsection{Wait-Until-Ready Policy}

The \emph{Wait-Until-Ready} policy adopts a more conservative strategy that avoids unnecessary regeneration when a partial entanglement link already exists. Its behavior is defined as follows:
\begin{itemize}
    \item If both elementary links $(1,2)$ and $(2,3)$ are present, an entanglement swap is attempted immediately, that is, $a^{sw}=1$.
    \item If only one elementary link exists, wait for the other elementary link.
    \item If neither elementary link exists, request entanglement generation on all visible links.
    \item If the elementary link age is higher than a typical small cut-off time, $\tau=4$, discard the existing link.
\end{itemize}

%$P_{12}=P_{23}=\begin{bmatrix}0.3 & 0.7 \\ 0.3 & 0.7\end{bmatrix}$

\subsection{Numerical Results}

Figure~\ref{fig:visible} presents the average AoE versus $p_L$ for mostly visible
links with $P_{12}=P_{23}=[0.3\ 0.7;\ 0.3\ 0.7]$. For all policies, the AoE
decreases as $p_L$ increases due to reduced waiting times for entanglement
generation. The RVI policy consistently achieves the lowest AoE across all
operating regimes, with the largest performance gains observed at low $p_L$,
where unreliable link generation amplifies the effects of entanglement aging.
The Greedy Gen+Swap ASAP policy performs comparably to RVI at high $p_L$ but
incurs higher AoE at low and moderate $p_L$ due to myopic swapping and
regeneration decisions. In contrast, the Wait-Until-Ready (WUR) policy
exhibits the highest AoE, as deferring swapping until both links are available
leads to excessive aging in quantum memory, particularly under intermittent
link availability.
\begin{figure}
    \centering
    \includegraphics[scale=0.65]{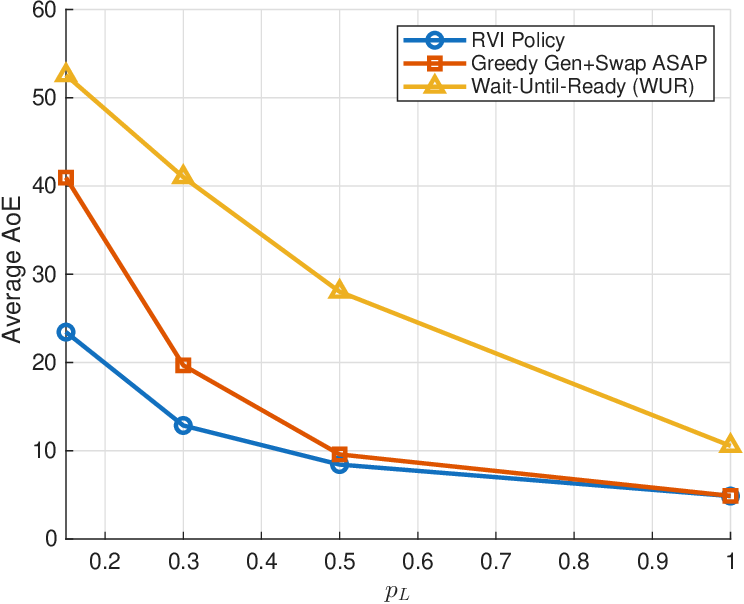}
    \caption{Average age of entanglement (AoE) versus the elementary link generation probability $p_L$ under mostly visible links with visibility transition matrices $P_{12}=P_{23}=[0.3\ 0.7;\ 0.3\ 0.7]$, when $p_{sw} = 0.8$.}
    \label{fig:visible}
\end{figure}

Figure~\ref{fig:moderate_visible} illustrates the same performance trends
under moderately visible links with $P_{12}=P_{23}=[0.6\ 0.4;\ 0.4\ 0.6]$.
While the average AoE again decreases with increasing $p_L$ for all policies,
the overall AoE levels are higher compared to the mostly visible case in
Fig.~\ref{fig:visible}, reflecting the reduced satellite availability.
Nevertheless, the RVI policy continues to outperform both baseline strategies
across all regimes, highlighting its robustness to degraded visibility
conditions.
\begin{figure}
    \centering
    \includegraphics[scale=0.65]{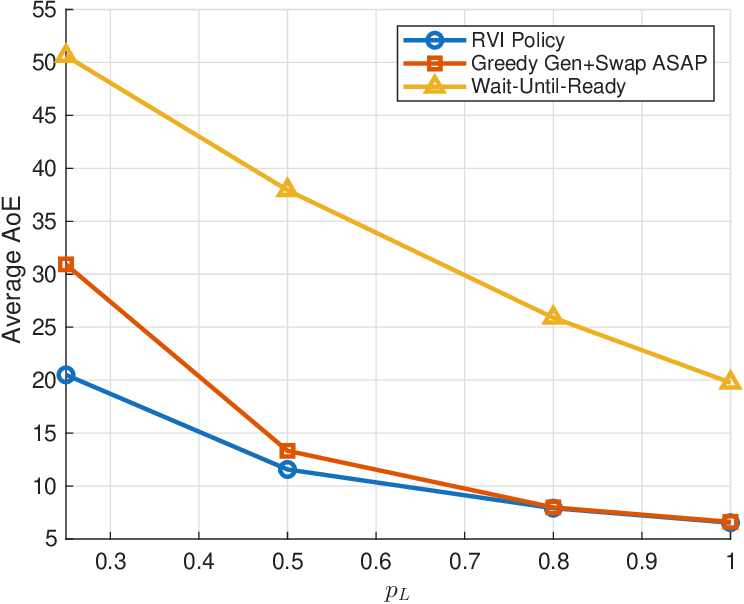}
    \caption{Average age of entanglement (AoE) versus elementary link generation probability $p_L$ under moderately visible links with visibility transition matrices $P_{12}=P_{23}=[0.6\ 0.4;\ 0.4\ 0.6]$, when $p_{sw} = 0.8$.}
    \label{fig:moderate_visible}
\end{figure}

Figure \ref{fig:wrtpsw} shows the average AoE as a function of the swap success probability
$p_{sw}$ for a fixed link generation probability $p_L = 0.3$ under asymmetric
visibility conditions. The visibility processes are given by
$P_{12}=[0.1\ 0.9;\ 0.1\ 0.9]$ and $P_{23}=[0.55\ 0.45;\ 0.55\ 0.45]$,
indicating that the first link is frequently visible while the second is more
intermittent. As $p_{sw}$ increases, the AoE decreases for all policies due to
reduced waiting in quantum memory. The RVI policy consistently achieves the
lowest AoE, while the Greedy policy degrades at low and moderate $p_{sw}$, and
the Wait-Until-Ready (WUR) policy incurs the highest AoE due to longer
storage under asymmetric link availability. Asymmetric link availability also degrades the optimal performance achieved by RVI algorithm. 

\begin{figure}
    \centering
    \includegraphics[scale=0.65]{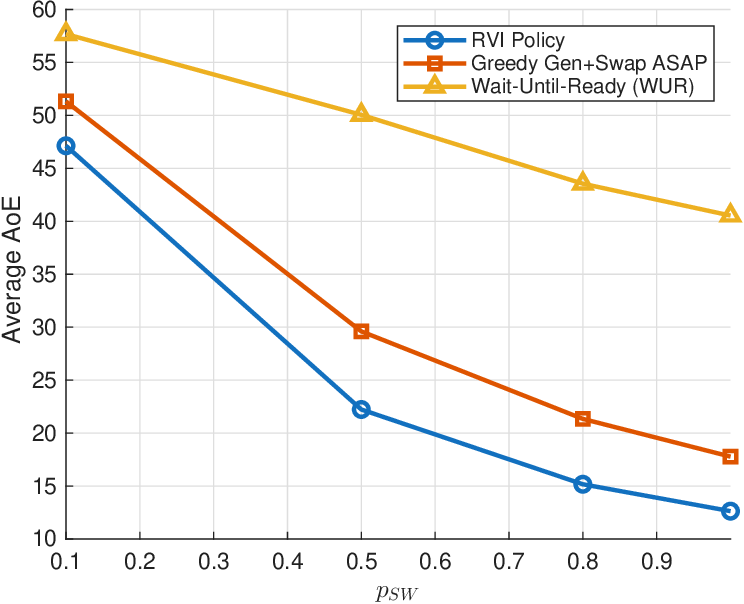}
    \caption{Average age of entanglement (AoE) versus swapping success probability $p_{sw}$ under asymmetric
visibility with visibility transition matrices $P_{12}=[0.1\ 0.9;\ 0.1\ 0.9]$ and $P_{23}=[0.55\ 0.45;\ 0.55\ 0.45]$, when $p_L = 0.3$.}
    \label{fig:wrtpsw}
\end{figure}

The convergence behavior of the proposed RVI algorithm is shown in Fig.~\ref{fig:covergence}. As $p_L$ and $p_{sw}$ decrease, convergence slows slightly but remains within a reasonable number of iterations.

\begin{figure}
    \centering
    \includegraphics[scale=0.65]{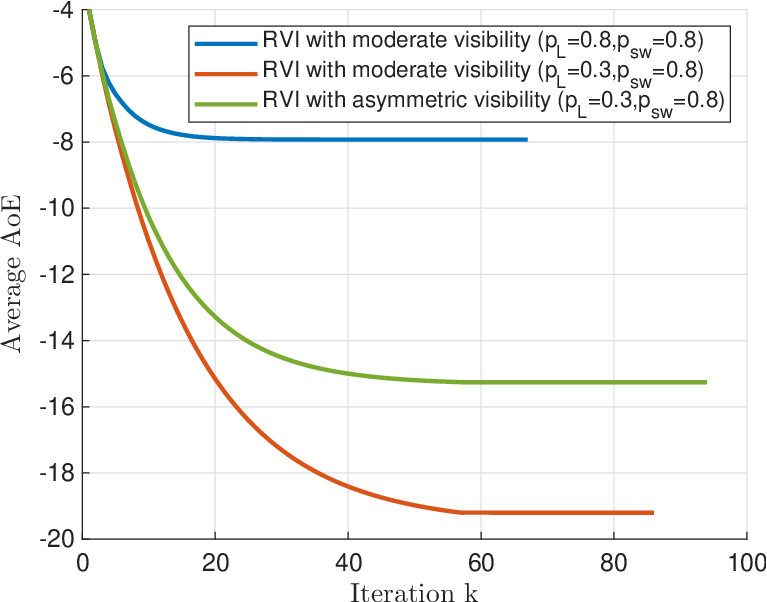}
    \caption{Convergence of the Relative Value Iteration (RVI) algorithm under different scenarios.}
    \label{fig:covergence}
\end{figure}

\section{Conclusion}

This paper has investigated the control of continuous entanglement distribution in
satellite-assisted quantum networks under probabilistic operations,
intermittent visibility, and quantum memory decoherence. By modeling
entanglement freshness using the novel AoE metric and
formulating the problem as an average-cost Markov decision process, we have demonstrated
that state-aware control policies can significantly outperform greedy and
conservative baseline strategies across a wide range of operating conditions.
The proposed RVI-based policy consistently achieves lower AoE by accounting
for future costs and heterogeneous link dynamics, particularly in regimes
characterized by low link reliability or asymmetric visibility. By keeping
entanglement fresh and optimally allocating actions, this work advances the
design of efficient and reliable quantum networking systems. This work also provides a baseline for future studies on AoE optimization in longer repeater chains and possibly the development of learning-based control algorithms to enable more scalable solutions in large-scale quantum
networks.


%\section*{Acknowledgment}




\bibliographystyle{IEEEtran}
%\begin{thebibliography}{1}


%\end{thebibliography}

\bibliography{Qunap.bib}

\end{document}